\documentclass[sigconf]{aamas}

\usepackage{balance} %
\usepackage{tikz}
\usepackage{algorithm}
\usepackage{algpseudocode}

\setcopyright{ifaamas}
\acmConference[AAMAS '24]{Proc.\@ of the 23rd International Conference
on Autonomous Agents and Multiagent Systems (AAMAS 2024)}{May 6 -- 10, 2024}
{Auckland, New Zealand}{N.~Alechina, V.~Dignum, M.~Dastani, J.S.~Sichman (eds.)}
\copyrightyear{2024}
\acmYear{2024}
\acmDOI{}
\acmPrice{}
\acmISBN{}

\acmSubmissionID{638}

\title{Symbolic Computation of Sequential Equilibria}

\author{Moritz Graf}
\affiliation{
  \institution{University of Freiburg}
  \city{Freiburg}
  \country{Germany}}
\email{grafm@cs.uni-freiburg.de}

\author{Thorsten Engesser}
\affiliation{
  \institution{IRIT}
  \city{Toulouse}
  \country{France}}
\email{thorsten.engesser@irit.fr}

\author{Bernhard Nebel}
\affiliation{
  \institution{University of Freiburg}
  \city{Freiburg}
  \country{Germany}}
\email{nebel@uni-freiburg.de}

\begin{abstract}

The sequential equilibrium is a standard solution concept for
ex\-ten\-si\-ve-form games with imperfect information that includes an explicit
representation of the players' beliefs. An assessment consisting of a strategy
and a belief is a sequential equilibrium if it satisfies the properties of
sequential rationality and consistency.

Our main result is that both properties %
together can be written as a single finite system of polynomial equations and
inequalities. The solutions to this system are exactly the sequential equilibria
of the game. We construct this system explicitly and describe an implementation
that solves it using cylindrical algebraic decomposition.
To write consistency as a finite system of equations, we need to compute the
extreme directions of a set of polyhedral cones. We propose a modified version
of the double description method, optimized for this specific purpose.
To the best of our knowledge, our implementation is the first to symbolically
solve general finite imperfect information games for sequential equilibria.%
\footnote{Our implementation is based on the open source Game
Theory Explorer \cite{gte} and can be downloaded at
\url{https://github.com/tengesser/GTE-sequential}.
}

\end{abstract}


\usepackage[bibliography=common]{apxproof}

\newtheoremrep{theorem}{Theorem}
\newtheoremrep{proposition}[theorem]{Proposition}
\newcommand{\citeay}[1]{\citeauthor{#1}~\cite{#1}}

\tikzstyle{thing} = [rectangle, minimum width=3cm, text width = 5cm, minimum height=1cm, text centered, draw=black, fill=gray!10]
\tikzstyle{arrow} = [thick,->,>=stealth]

\usetikzlibrary{shapes}
\usetikzlibrary{arrows.meta}
\newdimen\ndiam
\newdimen\sqwidth
\newdimen\spx
\newdimen\spy
\newdimen\yup
\newdimen\yfracup
\newdimen\paydown
\newdimen\treethickn
\newcommand{\twoeq}[2]{
\noindent
\begin{subequations}
\begin{minipage}{.49\linewidth}
\begin{align}
#1\\
#2
\end{align}
\end{minipage}
\end{subequations}}

\makeatletter
\gdef\@copyrightpermission{
	\begin{minipage}{0.3\columnwidth}
		\href{https://creativecommons.org/licenses/by/4.0/}{\includegraphics[width=0.90\textwidth]{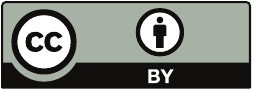}}
	\end{minipage}\hfill
	\begin{minipage}{0.7\columnwidth}
		\href{https://creativecommons.org/licenses/by/4.0/}{This work is licensed under a Creative Commons Attribution International 4.0 License.}
	\end{minipage}
	\vspace{5pt}
}
\makeatother

\begin{document}

\pagestyle{fancy}
\fancyhead{}

\maketitle

\section{Introduction}

The sequential equilibrium was proposed by \citeay{kreps1982sequential} as a
solution concept for extensive-form games with imperfect information.
In extensive-form games, players have multiple decision points in sequence, at
which they must decide how to act. A strategy for a player specifies what action
the player takes at each of their decision points.
When an extensive-form game has imperfect information, it means that the players
do not always know the exact state of the game. The reason is that there may be
actions of other players (or random events) that the players do not observe.

The sequential equilibrium is a generalization of subgame perfect equilibrium,
which is the standard solution concept for extensive-form games with perfect
information. In addition to a strategy, a sequential equilibrium specifies a set
of beliefs for each player, assigning probabilities to states that the player
cannot distinguish between.
Intuitively, two properties must be satisfied: strategies should be rational
given the players' beliefs (\emph{sequential rationality}), and beliefs should
be reasonable given the players' strategies (\emph{consistency}).
Having beliefs as an explicit part of the game's equilibria can provide
additional insight. The strategies specify what the players are doing, and the
beliefs provide an explanation why.

While finding all Nash equilibria or all subgame perfect equilibria of a game
has been implemented in tools like \textit{Gambit} \cite{gambit} or \textit{Game
Theory Explorer} \cite{gte}, there are no implemented solvers that symbolically
compute all sequential equilibria of a finite game. \citeay{azhar05} have
outlined an algorithm, enumerating so-called ``consistent bases'' of a game (of
which there might be exponentially many) and characterizing the sequential
equilibria for each basis by a system of polynomial equations and inequalities.

We take a similar approach, characterizing all equilibria of the game by a
single such system. To this end, we combine existing results from the literature
\cite{hendon1996one,kohlberg1997independence}. Most importantly, our paper
details all the steps necessary to generate and solve the system.

Polynomial systems of equations can in general have an infinite number of
solutions, and indeed games with imperfect information often have an infinite
number of sequential equilibria. Therefore, solving such a system does not
require simply enumerating all the solutions, but rather finding a description
of the solutions that allows any one of them to be easily extracted. For this
purpose we will use the cylindrical algebraic decomposition algorithm provided
by the computer algebra system \textit{Mathematica}. We obtain a list of
intervals, one for each of the variables, which are stratified in the sense that
the boundaries of the intervals depend only on the variables before them. With
this, any sequential equilibria of the game can be obtained by successively
choosing a value for each variable.
Using symbolic computation has the added advantage of allowing us to solve
families of games where some outcomes or probabilities of random events are
controlled by a set of parameters.

In Section~\ref{sec:gametheory} we recapitulate the basic definitions and
solution concepts for extensive-form games.
In Section~\ref{sec:sequential} we study the properties of sequential equilibria
and show how both sequential rationality and consistency can be written as a
system of polynomial equations and inequalities.
Section~\ref{sec:implementation} shows the steps required to implement a
sequential equilibrium solver, as well as some strategies for reducing
computation time.
We conclude in Section~\ref{sec:conclusion}.

\section{Theoretical Background}\label{sec:gametheory}

In Section~\ref{sec:notation}, we introduce the notation used throughout this
paper, mostly following \citeay{osborne1994course}.
In Section~\ref{sec:concepts} we then recapitulate the most important solution
concepts, leading to the definition of sequential equilibria.

\subsection{Extensive-Form Games}\label{sec:notation}

\subsubsection*{Game}

A game consists of a set of \emph{players} $N = \{1, \ldots, n\}$, a set of
\emph{histories} $H$ of which a subset $Z \subseteq H$ are \emph{terminal}
histories, a \emph{player function} $N(h)$ that assigns an acting player to each
non-terminal history, a function $A(h)$ specifying the set of \emph{actions}
available at each non-terminal history, and a \emph{utility function} $u_i(h^*)$
that assigns to each player $i \in N$ a utility for each of the terminal
histories $h^* \in Z$. Imperfect information is represented by a set of
\emph{information sets} $\mathcal{I}$ that partitions the set of non-terminal
histories. We consider only games where the set of histories $H$ is finite.

\subsubsection*{Actions and Histories}

Histories can be thought of as nodes in a game tree. Each history $h = \langle
a_1, \ldots, a_k \rangle \in H$ encodes the sequence of actions leading to that
node.
We say that $h = \langle a_1, \ldots, a_k \rangle$ is a \emph{prefix}\ of $h' =
\langle a'_1, \ldots, a'_l \rangle$ if $k\leq l$ and $a'_i = a_i$ for all $i \in
\{1, \ldots, k\}$ (that is, if $h = h'$ or if $h$ is an ascendant of $h'$ in the
game tree). If $h \in H$, it must thus be the case that $h' \in H$ for all
prefixes $h'$ of $h$.
The terminal histories $h^* \in Z$ are exactly those histories which are not
prefix of some other $h' \neq h^*$ from $H$ (the leaf nodes of the tree).

\subsubsection*{Information Sets}

The information sets $I \in \mathcal{I}$ are used to represent imperfect
information in the game. After some history $h \in I$ is reached, the acting
player cannot distinguish whether $h$ or any other history $h' \in I$ is the
actual history. This implies that in each history of an information set both the
acting player and the set of available actions must be identical. We thus
usually write $N(I)$ and $A(I)$ instead of $N(h)$ and $A(h)$.
We say that a game has \emph{perfect recall} if for any two histories $h, h'$ in
the same information set $I$ of player $i$, the sequence of information sets
encountered and actions played by $i$ are identical for both $h$ and $h'$. This
poses a restriction on game trees: players cannot `forget' information about
their position in the game tree or about the actions they played. Sequential
equilibria are only defined for games with perfect recall.

\subsubsection*{Strategies and Beliefs}
Solution concepts usually contain a strategy profile $\beta = (\beta_1,
\beta_2, \ldots, \beta_n)$, consisting of a behavioral strategy $\beta_i$ for each
player $i$.
These $\beta_i$ assign to each information set $I$ of player $i$ a probability
distribution over all the actions $a\in A(I)$, such that player $i$ plays action
$a\in A(I)$ at $I$ with probability $\beta_i(I)(a)$.

Similarly, a system of beliefs $\mu$ assigns to each information set $I$ a
probability distribution over all histories $h\in I$, such that player $i$
believes to be in history $h \in I$ with probability $\mu_i(I)(h)$.

It is often not important which player a strategy or belief belongs to. Since
each information set has a distinct acting player $i$, we will often write
$\beta(I)(a)$ and $\mu(I)(h)$ without the explicit subscript $i$.

Based on the players strategies $\beta$, we can formulate the probabilities of
reaching a specific history $h = \langle a_1, \ldots, a_k\rangle$ as
\[P_\beta(h) = \prod_{i = 1}^k \beta(I_i)(a_i),\]
where $I_i$ is the information set containing $\langle a_1, \ldots, a_{i-1}
\rangle$.
We can expand this to the probability of reaching an information set $I$:
\[P_\beta(I) = \sum_{h\in I} P_\beta(h)\]
We will also use the conditional probabilities of reaching history $h' = \langle
a_1, \ldots, a_k\rangle$ starting from a prefix $h = \langle a_1, \ldots,
a_l\rangle$, $l<k$:
\begin{align*}
P_\beta(h' | h)
 &= \frac{P_\beta(h')}{P_\beta(h)} %
 = \frac{\prod_{i = 1}^k \beta(I_i)(a_i)}{\prod_{i = 1}^l \beta(I_i)(a_i)}
 = \prod_{i = l+1}^k \beta(I_i)(a_i)
\end{align*}
Note that $P_\beta(h|h') = 1$, since $h'$ can only be reached after $h$ is reached. If $h$ is not a prefix of $h'$,
then $P_\beta(h' | h) =0$.

\subsubsection*{Utilities}

From the utilities $u_i(h^*)$ assigned to terminal histories, we define player
$i$'s expected utility given strategy profile $\beta$:
\[U^E_i(\beta) = \sum_{h^* \in Z} u_i(h^*) P_\beta(h^*)\]
We generalize this to player $i$'s expected utility of the subgame starting at
$h$:
\[U^E_i(\beta | h) = \sum_{h^* \in Z} u_i(h^*) P_\beta(h^* | h)\]
To obtain the utility that player $i$ assigns to an information set $I$, we need
the concept of believed utility:
\[U^B_i(\beta, \mu | I) = \sum_{h^* \in Z} u_i(h^*)
  \sum_{h\in I} \mu(I)(h) P_\beta(h^* | h)\]

Note that at most one term of the inner sum is nonzero, because for every terminal history $h^*$ there is at most one $h \in I$
such that $h$ is a prefix of $h^*$ and therefore $P_\beta(h^* | h) \neq 0$.

\subsection{Solution Concepts}\label{sec:concepts}

The most fundamental solution concept in game theory is the Nash equilibrium.
Intuitively, a strategy profile is a Nash equilibrium if no player can improve
their utility by deviating from their strategy.

\begin{definition}[Nash equilibrium]\label{def:nash}
A strategy profile $\beta$ is a Nash equilibrium if
$U^E_i(\beta') \leq U^E_i(\beta)$
for all players $i$ and strategy profiles $\beta'=( \alpha_i, \beta_{-i}
)$ which only deviate from $\beta$ in player $i$'s strategy.
\end{definition}

For sequential games, where players do not take their actions simultaneously,
Nash equilibria can run into the problem of \textit{non-credible threats}. An example
is shown in Figure~\ref{fig:falsethreat}. The highlighted strategy profile,
where player 1 plays $b$ and player 2 plays $x$, is a Nash equilibrium.
However, playing $x$ is irrational for player $2$: If history $\left< a \right>$
was reached, player $2$ should play $y$ to get a higher payoff.
This is not captured by Nash equilibria. Its definition considers only the
expected utility of the whole game, which does not change when strategies in
unreached parts of the game tree change.

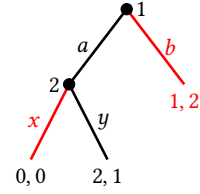
\begin{figure}
  \centering
  \begin{tikzpicture}[scale=0.5, StealthFill/.tip={Stealth[line width=.7pt,inset=0pt,length=13pt,angle'=30]}]
\draw [line width=\treethickn] (0,-0) node[right] {1};
\draw [line width=\treethickn] (-1.53,-2) node[left] {2}
   to node[left] {$a$} (0,-0);
\draw [line width=\treethickn, red] (1.53,-2)
   node[below] {$1,2$}
   to node[right] {\textcolor{red}{$b$}} (0,-0);
\draw [line width=\treethickn, red] (-2.55,-4)
   node[below,black] {$0,0$}
   to node[left] {\textcolor{red}{$x$}} (-1.53,-2);
\draw [line width=\treethickn] (-0.51,-4)
   node[below] {$2,1$}
   to  node[right] {$y$} (-1.53,-2);
\node[inner sep=0pt,minimum size=\ndiam, draw, fill, shape=circle] at (-1.53,-2) {};
\node[inner sep=0pt,minimum size=\ndiam, draw, fill, shape=circle] at (0,-0) {};
\end{tikzpicture}
  \caption{A non-credible threat in an extensive-form game. The actions that are played are marked in red. Each terminal node is annotated with the utilities of player 1 and 2.}
  \label{fig:falsethreat}
\end{figure}

This leads to the subgame perfect equilibrium, a refinement of the Nash
equilibrium in which deviations from the strategy profile must also not increase
the payoffs received by players in any subgame.

\begin{definition}[Subgame Perfect Equilibrium]\label{def:spe} A strategy
profile $\beta$ is a subgame perfect equilibrium if
$U^E_i(\beta' | h) \leq U^E_i(\beta | h)$
for all players $i$, all histories $h$ where that player acts, and all strategy
profiles $\beta'=( \alpha_i, \beta_{-i} )$ that only deviate from
$\beta$ in player $i$'s strategy.
\end{definition}

The marked strategy profile $\beta$ in Figure~\ref{fig:falsethreat} is not
subgame perfect, since
$U^E_2(\beta' | \langle a \rangle) = 1 > 0 = U^E_2(\beta | \langle a \rangle)$,
where $\beta'$ is the strategy profile in which player 2 plays $y$ at $\langle a
\rangle$ instead of $x$.

In imperfect information games, when players reach some information set
$I$, they must decide which action to play without knowing which $h \in I$ is
the actual history. To preserve the concept of subgame perfection, players must
form a belief $\mu(I)(h)$ about which history they are in. Pairs $(\beta, \mu)$
of a strategy profile $\beta$ and a belief system $\mu$ are called
\emph{assessments}.
Intuitively, strategies must be believed to be optimal from each information
set, and beliefs must reflect the probabilities of histories being reached based
on the agents' strategies. This leads us to the following definition.

\begin{definition}[Sequential Equilibrium]
Let $\beta$ be a strategy profile and $\mu$ a belief system. The assessment
$(\beta, \mu)$ is
\begin{itemize}
    \item[(i)] \emph{sequentially rational} if for any information set $I$ of
    player $i$ and any strategy profile $\beta' = (\alpha_i, \beta_{-i})$ with a
    different strategy for player $i$, we have
    \begin{align}
    U^B_i(\beta', \mu | I) \leq U^B_i(\beta, \mu | I),
    \end{align}
    \item[(ii)] \emph{consistent} if there exists a series of assessments
    $(\beta^n, \mu^n)$ such that $\lim_{n \to \infty} (\beta^n, \mu^n) = (\beta,
    \mu)$, and for all $h\in H$ and $I\in \mathcal{I}$,
    \begin{align}
    P_{\beta^n}(h) &> 0 \\
    \mu^n(I)(h) &= \frac{P_{\beta^n}(h)}{P_{\beta^n}(I)}\text{,}
    \end{align}
    \item[(iii)] a \emph{sequential equilibrium} if it is both sequentially
    rational and consistent.
\end{itemize}
\end{definition}

\section{Exploring Sequential Equilibria}\label{sec:sequential}

In the following section we will explore these two properties of sequential
equilibria. Since our goal is to compute all sequential equilibria of a given
game, we will attempt to find ways to transform both properties into a form that
can be used computationally.

For sequential rationality, we will reduce the problem to a more local one,
where we only need to consider alternative strategies $\beta'$ that differ from
$\beta$ only at a single information set $I$, and we will end up with sequential
rationality described by a system of equations and inequalities that are linear
in the beliefs $\mu(I)(h)$ and polynomial in the action probabilities
$\beta(I)(a)$. This largely follows the work of \citeay{hendon1996one} and
requires that the beliefs are consistent.

We then follow the work of \citeay{kohlberg1997independence} to also express
consistency as a finite set of polynomial equations. This requires us to compute
the extreme directions of a set of polyhedral cones, a problem we will discuss
further in Section~\ref{sec:cones}.

\begin{figure}
    \begin{tikzpicture}[scale=.5, StealthFill/.tip={Stealth[line width=.7pt,inset=0pt,length=13pt,angle'=30]}]
\draw [] (-3.32,-1.7) arc(90:270:0.3) -- (3.32,-2.3) arc(-90:90:0.3) -- cycle;
\draw [] (-2.04,-3.7) arc(90:270:0.3) -- (2.04,-4.3) arc(-90:90:0.3) -- cycle;
\draw [line width=\treethickn] (0,-0) node[above] {$1$};
\draw [line width=\treethickn, red] (-3.32,-2)
   node[right] {\textcolor{blue}{1}}
   to node[above] {$a$} (0,-0);
\draw [line width=\treethickn] (3.32,-2)
   node[left] {\textcolor{blue}{0}}
   to node[above] {$b$} (0,-0);
\draw [line width=\treethickn] (-5.11,-4)
   node[below,yshift=0.1\paydown] {$0,0$}
   to node[left] {$c$} (-3.32,-2);
\draw [line width=\treethickn, red] (-2.04,-4)
   node[right] {\textcolor{blue}{0}}
   to node[right] {$d$} (-3.32,-2);
\draw [line width=\treethickn, red] (2.04,-4)
   node[left] {\textcolor{blue}{1}}
   to node[left] {$d$} (3.32,-2);
\draw [line width=\treethickn] (5.11,-4)
   node[below,yshift=0.1\paydown] {$0,1$}
   to node[right] {$c$} (3.32,-2);
\draw [line width=\treethickn] (-3.06,-6)
   node[below,yshift=0.1\paydown] {$0,2$}
   to node[left] {$e$} (-2.04,-4);
\draw [line width=\treethickn, red] (-1.02,-6)
   node[below,yshift=0.1\paydown] {$0,1$}
   to node[right] {$f$} (-2.04,-4);
\draw [line width=\treethickn] (1.02,-6)
   node[below,yshift=0.1\paydown] {$0,0$}
   to node[left] {$e$} (2.04,-4);
\draw [line width=\treethickn, red] (3.06,-6)
   node[below,yshift=0.1\paydown] {$0,2$}
   to node[right] {$f$} (2.04,-4);
\draw (0,-2) node[xshift=0.0cm] {$2$};
\draw (0,-4) node[xshift=0.0cm] {$2$};
\node[inner sep=0pt,minimum size=\ndiam, draw, fill, shape=circle] at (2.04,-4) {};
\node[inner sep=0pt,minimum size=\ndiam, draw, fill, shape=circle] at (-2.04,-4) {};
\node[inner sep=0pt,minimum size=\ndiam, draw, fill ,shape=circle] at (0,-0) {};
\node[inner sep=0pt,minimum size=\ndiam, draw, fill, shape=circle] at (-3.32,-2) {};
\node[inner sep=0pt,minimum size=\ndiam, draw, fill, shape=circle] at (3.32,-2) {};
\end{tikzpicture}
    \caption{A game with an inconsistent assessment. Beliefs are depicted in blue, strategies in red. Although the assessment is \emph{locally} sequentially rational, it is not sequentially rational.}
    \label{fig:localsr}
\end{figure}
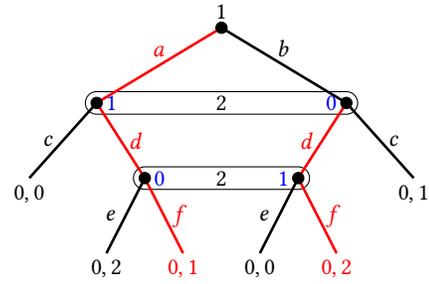

\subsection{Sequential Rationality}

Sequential rationality is the natural extension of subgame perfection. At every
information set, the acting player must believe that no deviating strategy can
improve their utility. That is, the acting player cannot achieve a higher payoff
by changing their own action probabilities in that information set or further
down the game tree.

We recapitulate the result of \citeay{hendon1996one}, that if beliefs are
consistent, we can reduce sequential rationality to a more local property. We
also show how this property can be described by a set of polynomial equations
and inequalities in $\beta$ and $\mu$.
These results are similar to two concepts in the analysis of Nash equilibria and
subgame perfect equilibria: the one-shot deviation principle, which states that
we only need to consider local deviations to verify subgame perfection, and the
fact that players may only assign positive probabilities to actions which are
best responses to the opponents' strategy.

\begin{toappendix}
The believed utility of player $i$ at a given information set $I$, for an
assessment $(\beta, \mu)$, is given by
\[U^B_i(\beta,\mu | I)= \sum_{h^* \in Z} u_i(h^*) \sum_{h \in I} \mu(I)(h) \cdot
    P_\beta(h^*| h).\]
It is helpful to write the believed utility at $I$ in terms of the believed
utility of playing each action $a\in A(I)$:
\begin{align*}
    U^B_i(\beta,\mu | I) &  = \sum_{h^* \in Z} u_i(h^*) \sum_{h \in I} \mu(I)(h) \cdot  P_\beta(h^*| h) \\
    & = \sum_{h^* \in Z} u_i(h^*) \sum_{h \in I} \mu(I)(h) \sum_{a\in A(I)} \beta(I)(a) \cdot P_\beta(h^*|
    \langle h, a \rangle)\\
    & = \sum_{a\in A(I)} \beta(I)(a) \sum_{h \in I} \mu(I)(h) \sum_{h^* \in Z} u_i(h^*) \cdot P_\beta (h^*| \langle
    h, a \rangle)\\
    & = \sum_{a\in A(I)} \beta(I)(a) \sum_{h \in I} \mu(I)(h) U^E_i(\beta | \langle h, a \rangle)
\end{align*}
Note that $P_\beta(h^*| \langle h, a \rangle) = 0$ for all $a$ where $\langle h,
a \rangle$ is not a prefix of $h^*$. For ease of notation, we define
\[U^B_i(\beta, \mu, | I,a) = \sum_{h \in I} \mu(I)(h) U^E_i(\beta | \langle h, a \rangle)\]
as the believed utility player $i$ assigns to playing $a$ at $I$. We thus get
\[U^B_i(\beta,\mu | I) = \sum_{a\in A(I)} \beta(I)(a) U^B_i(\beta, \mu, | I,a).\]
\end{toappendix}

\subsubsection{One-Shot Deviation Principle}
Since sequential rationality is described as a property that holds for each
information set, it is natural to define \emph{sequential rationality at $I$} to
mean that this property holds for a given information set $I$. In addition, we
define \emph{local sequential rationality at $I$} to mean that the property of
sequential rationality at $I$ holds for all strategies that differ from $\beta$
only at $I$. These concepts are used in \cite{hendon1996one}, but not defined by
these names.

\begin{definition}[Local Sequential Rationality]\label{def:localsr}
Let $(\beta, \mu)$ be an assessment and $I\in \mathcal{I}$ an information set of
acting player $i$. The assessment $(\beta, \mu)$ is \emph{locally sequentially
rational at $I$} if $U^B_i(\beta', \mu| I) \leq U^B_i(\beta, \mu| I)$ for any
strategy profile $\beta' = (\alpha_i, \beta_{-i})$ where agent $i$ plays some
strategy $\alpha_i$, with $\alpha_i(I') = \beta_i(I')$ for all $I' \in
\mathcal{I} \setminus \{I\}$. We say $(\beta, \mu)$ is \emph{locally
sequentially rational} if it is locally sequentially rational at every
information set $I \in \mathcal{I}$.
\end{definition}

We restate the one-shot deviation principle from \citeay{hendon1996one}:

\begin{theorem}[One-Shot Deviation Principle, \cite{hendon1996one}]\label{thrm:oneshot}
Let $(\beta, \mu)$ be a locally sequentially rational assessment.
If $(\beta, \mu)$ is consistent, then $(\beta, \mu)$ is sequentially rational and therefore a sequential
equilibrium.
\end{theorem}

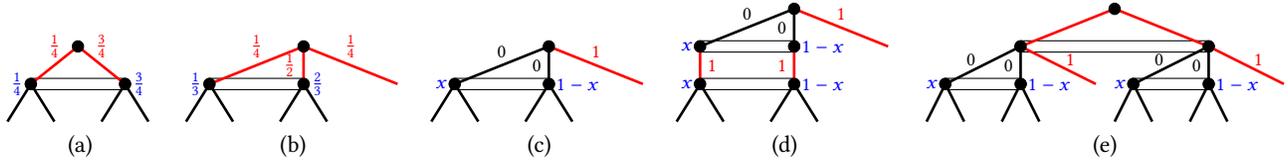
\begin{figure*}[ht]
    \begin{tabular}{ccccc}
        \footnotesize \begin{tikzpicture}[scale=.25 , StealthFill/.tip={Stealth[line width=.7pt,inset=0pt,length=13pt,angle'=30]}]
\draw [] (-2.5,-1.7) arc(90:270:0.3) -- (2.5,-2.3) arc(-90:90:0.3) -- cycle;
\draw [line width=\treethickn] (0,-0);
\draw [line width=\treethickn, red] (-2.5,-2)
   node[left] {\textcolor{blue}{$\frac{1}{4}$}}
   to node[above] {\textcolor{red}{$\frac{1}{4}$}} (0,-0);
\draw [line width=\treethickn, red] (2.5,-2)
   node[right] {\textcolor{blue}{$\frac{3}{4}$}}
   to node[above] {\textcolor{red}{$\frac{3}{4}$}} (0,-0);
\draw [line width=\treethickn] (-3.75,-4) -- (-2.5,-2);
\draw [line width=\treethickn] (-1.25,-4) -- (-2.5,-2);
\draw [line width=\treethickn] (1.25,-4) -- (2.5,-2);
\draw [line width=\treethickn] (3.75,-4) -- (2.5,-2);
\node[inner sep=0pt,minimum size=\ndiam, draw, fill, shape=circle] at (0,-0) {};
\node[inner sep=0pt,minimum size=\ndiam, draw, fill, shape=circle] at (-2.5,-2) {};
\node[inner sep=0pt,minimum size=\ndiam, draw, fill, shape=circle] at (2.5,-2) {};
\end{tikzpicture} &
        \footnotesize \begin{tikzpicture}[scale=.25, StealthFill/.tip={Stealth[line width=.7pt,inset=0pt,length=13pt,angle'=30]}]
\draw [] (-5.00,-1.7) arc(90:270:0.3) -- (0.0,-2.3) arc(-90:90:0.3) -- cycle;
\draw [line width=\treethickn] (0,-0);
\draw [line width=\treethickn, red] (-5.00,-2)
    node[left] {\textcolor{blue}{$\frac{1}{3}$}}
   to node[above] {\textcolor{red}{$\frac{1}{4}$}} (0,-0);
\draw [line width=\treethickn, red] (0.0,-2)
    node[right] {\textcolor{blue}{$\frac{2}{3}$}}
    to node[left] {\textcolor{red}{$\frac{1}{2}$}} (0,-0);
\draw [line width=\treethickn, red] (5.00,-2) to
  node[above] {\textcolor{red}{$\frac{1}{4}$}} (0,-0);
\draw [line width=\treethickn] (-6.25,-4) -- (-5.00,-2);
\draw [line width=\treethickn] (-3.75,-4) -- (-5.00,-2);
\draw [line width=\treethickn] (-1.25,-4) -- (0.0,-2);
\draw [line width=\treethickn] (1.25,-4) -- (0.0,-2);
\node[inner sep=0pt,minimum size=\ndiam, draw, fill, shape=circle] at (0,-0) {};
\node[inner sep=0pt,minimum size=\ndiam, draw, fill, shape=circle] at (-5.00,-2) {};
\node[inner sep=0pt,minimum size=\ndiam, draw, fill, shape=circle] at (0.0,-2) {};
\end{tikzpicture} &
        \footnotesize \begin{tikzpicture}[scale=.25, StealthFill/.tip={Stealth[line width=.7pt,inset=0pt,length=13pt,angle'=30]}]
\draw [] (-5.00,-1.7) arc(90:270:0.3) -- (0.00,-2.3) arc(-90:90:0.3) -- cycle;
\draw [line width=\treethickn] (0,-0);
\draw [line width=\treethickn] (-5.00,-2)
   node[left] {\textcolor{blue}{$x$}}
   to node[above] {$0$} (0,-0);
\draw [line width=\treethickn] (0.00,-2)
   node[right] {\textcolor{blue}{$1-x$}}
   to node[left] {$0$} (0,-0);
\draw [line width=\treethickn, red] (5.00,-2)
   to node[above] {\textcolor{red}{$1$}} (0,-0);
\draw [line width=\treethickn] (-6.25,-4) -- (-5.00,-2);
\draw [line width=\treethickn] (-3.75,-4) -- (-5.00,-2);
\draw [line width=\treethickn] (-1.25,-4) -- (0.00,-2);
\draw [line width=\treethickn] (1.25,-4) -- (0.00,-2);
\node[inner sep=0pt,minimum size=\ndiam, draw, fill, shape=circle] at (0,-0) {};
\node[inner sep=0pt,minimum size=\ndiam, draw, fill, shape=circle] at (-5.00,-2) {};
\node[inner sep=0pt,minimum size=\ndiam, draw, fill, shape=circle] at (0.00,-2) {};
\end{tikzpicture} &
        \footnotesize \begin{tikzpicture}[scale=.25, StealthFill/.tip={Stealth[line width=.7pt,inset=0pt,length=13pt,angle'=30]}]
\draw [] (-5.00,-1.7) arc(90:270:0.3) -- (0.00,-2.3) arc(-90:90:0.3) -- cycle;
\draw [] (-5.00,-3.7) arc(90:270:0.3) -- (0.00,-4.3) arc(-90:90:0.3) -- cycle;
\draw [line width=\treethickn] (0,-0);
\draw [line width=\treethickn] (-5.00,-2)
    node[left] {\textcolor{blue}{$x$}}
   to node[above]{$0$} (0,-0);
\draw [line width=\treethickn] (0.00,-2)
    node[right] {\textcolor{blue}{$1-x$}}
   to node[left] {$0$} (0,-0);
\draw [line width=\treethickn, red] (5.00,-2)
   to node[above] {\textcolor{red}{$1$}} (0,-0);
\draw [line width=\treethickn, red] (-5.00,-4)
    node[left] {\textcolor{blue}{$x$}}
   to node[right] {\textcolor{red}{$1$}} (-5.00,-2);
\draw [line width=\treethickn, red] (0.00,-4)
    node[right] {\textcolor{blue}{$1-x$}}
   to node[left] {\textcolor{red}{$1$}} (0.00,-2);
\draw [line width=\treethickn] (-6.25,-6) -- (-5.00,-4);
\draw [line width=\treethickn] (-3.75,-6) -- (-5.00,-4);
\draw [line width=\treethickn] (-1.25,-6) -- (0.00,-4);
\draw [line width=\treethickn] (1.25,-6) -- (0.00,-4);
\node[inner sep=0pt,minimum size=\ndiam, draw, fill, shape=circle] at (0.00,-4) {};
\node[inner sep=0pt,minimum size=\ndiam, draw, fill, shape=circle] at (-5.00,-4) {};
\node[inner sep=0pt,minimum size=\ndiam, draw, fill, shape=circle] at (0,-0) {};
\node[inner sep=0pt,minimum size=\ndiam, draw, fill, shape=circle] at (-5.00,-2) {};
\node[inner sep=0pt,minimum size=\ndiam, draw, fill, shape=circle] at (0.00,-2) {};
\end{tikzpicture} &
        \footnotesize \begin{tikzpicture}[scale=.25, StealthFill/.tip={Stealth[line width=.7pt,inset=0pt,length=13pt,angle'=30]}]
\draw [] (-5,-1.7) arc(90:270:0.3) -- (5,-2.3) arc(-90:90:0.3) -- cycle;
\draw [] (-9,-3.7) arc(90:270:0.3) -- (-5,-4.3) arc(-90:90:0.3) -- cycle;
\draw [] (1,-3.7) arc(90:270:0.3) -- (5,-4.3) arc(-90:90:0.3) -- cycle;
\draw [line width=\treethickn] (0,0);
\draw [line width=\treethickn, red] (-5,-2) -- (0,-0);
\draw [line width=\treethickn, red] (5,-2) -- (0,-0);
\draw [line width=\treethickn] (-9,-4)
    node[left] {\textcolor{blue}{$x$}}
   to node[left,yshift=2pt] {$0$} (-5,-2);
\draw [line width=\treethickn] (-5,-4)
   node[right] {\textcolor{blue}{$1-x$}}
   to node[left] {$0$} (-5,-2);
\draw [line width=\treethickn] (1,-4)
    node[left] {\textcolor{blue}{$x$}}
   to node[left,yshift=2pt] {$0$} (5,-2);
\draw [line width=\treethickn] (5,-4)
    node[right] {\textcolor{blue}{$1-x$}}
   to node[left] {$0$} (5,-2);
\draw [line width=\treethickn, red] (-1,-4)
   to node[right,yshift=2pt] {\textcolor{red}{$1$}} (-5,-2);
\draw [line width=\treethickn, red] (9,-4)
   to node[right,yshift=2pt] {\textcolor{red}{$1$}} (5,-2);
\draw [line width=\treethickn] (-10,-6) -- (-9,-4);
\draw [line width=\treethickn] (-8,-6) -- (-9,-4);
\draw [line width=\treethickn] (-6,-6) -- (-5,-4);
\draw [line width=\treethickn] (-4,-6) -- (-5,-4);
\draw [line width=\treethickn] (0,-6) -- (1,-4);
\draw [line width=\treethickn] (2,-6) -- (1,-4);
\draw [line width=\treethickn] (4,-6) -- (5,-4);
\draw [line width=\treethickn] (6,-6) -- (5,-4);
\node[inner sep=0pt,minimum size=\ndiam, draw, fill, shape=circle] at (1,-4) {};
\node[inner sep=0pt,minimum size=\ndiam, draw, fill, shape=circle] at (-5,-4) {};
\node[inner sep=0pt,minimum size=\ndiam, draw, fill, shape=circle] at (-9,-4) {};
\node[inner sep=0pt,minimum size=\ndiam, draw, fill, shape=circle] at (0,0) {};
\node[inner sep=0pt,minimum size=\ndiam, draw, fill, shape=circle] at (5,-4) {};
\node[inner sep=0pt,minimum size=\ndiam, draw, fill, shape=circle] at (-5,-2) {};
\node[inner sep=0pt,minimum size=\ndiam, draw, fill, shape=circle] at (5,-2) {};
\end{tikzpicture} \\
        (a) & (b) & (c) & (d) & (e)
    \end{tabular}
    \caption{Five examples of game trees annotated with consistent assessments.
    Beliefs are marked in blue, strategies in red.}
    \label{fig:consistency}
\end{figure*}

Since sequential rationality implies local sequential rationality, we can follow
from Theorem~\ref{thrm:oneshot} that whenever an assessment is consistent, it is
a sequential equilibrium if and only if it is also locally sequentially
rational.
Note that while our version of the theorem requires that the assessment be
consistent, \citeay{hendon1996one} use a weaker form of consistency called
\textit{pre-consistency} that is sufficient for local sequential rationality to
imply sequential rationality. \citeay{perea2002note} has further reduced the
requirement to so-called \textit{updating consistency}.
Since consistency is a necessary condition for sequential equilibria, we
will not use these weaker concepts.

Figure~\ref{fig:localsr} shows why, without a consistency requirement, local sequential
rationality is not a sufficient condition for sequential rationality.
The problem is that even if an assessment is locally sequentially rational, a
deviation in one information set might be believed to be profitable given the
beliefs at an earlier information set. Consider the strategy where agent $2$
plays $e$ instead of $f$ in the lower information set. While the believed
utility in this information set decreases from $2$ to $0$, the believed utility
at the upper information set increases from $1$ to $2$. Thus, while the depicted
assessment is locally sequentially rational, it is not sequentially rational.

\subsubsection{Best Responses}
The following proposition provides a necessary and sufficient condition on local
sequential rationality. If an action $a$ is played with probability
$\beta(I)(a)>0$, then it must be a best response to the other players' actions.
Our proposition generalizes Lemma~33.2 from the book by
\citeauthor{osborne1994course}.

\begin{propositionrep}\label{prop:support}
An assessment $(\beta, \mu)$ is locally sequentially rational if and only if for
all $I \in \mathcal{I}$ and $a \in A(I)$ the following holds:
\begin{align}
\text{ if } \beta(I)(a) > 0, \text{ then } \sum_{h \in I} \mu(I)(h) U^E_i(\beta | \langle h, a \rangle) = U^B_i (\beta,\mu | I) \label{eq:sup1}\\
\text{ if } \beta(I)(a) = 0, \text{ then } \sum_{h \in I} \mu(I)(h) U^E_i(\beta | \langle h, a \rangle) \leq U^B_i(\beta,\mu | I) \label{eq:sup2}
\end{align}
\end{propositionrep}

\begin{proofsketch}
We apply the proof idea for best responses in strategic games: If there is an
action that violates (\ref{eq:sup1}) or (\ref{eq:sup2}), then we can construct a
local deviation of $\beta$ that results in a higher utility, violating local
sequential rationality. If the assessment is not locally sequentially rational,
then there must exist a local deviation with higher utility, which is only
possible if (\ref{eq:sup1}) or (\ref{eq:sup2}) are violated.
\end{proofsketch}

\begin{proof}
``$\Rightarrow$''
Assume $(\beta, \mu)$ is locally sequentially rational and Property
(\ref{eq:sup1}) does not hold for some $I\in \mathcal{I}, a_1\in A(I)$, $\beta(I)(a_1)>0$. Since
\[U^B_i (\beta,\mu | I) = \sum_{a\in A(I)} \beta(I)(a) \sum_{h \in I} \mu (I)(h)
U^E_i(\beta | \langle h, a \rangle),\]
there must be some $a_2\in A(I), a_2\neq a_1$ with $\beta(I)(a_2) > 0$ where Property
(\ref{eq:sup1}) also does not hold, such that without loss of generality:
\[\sum_{h \in I} \mu(I)(h) U^E_i(\beta | \langle h, a_1 \rangle) > U^B_i
(\beta,\mu | I) > \sum_{h \in I} \mu(I)(h) U^E_i(\beta | \langle h, a_2
\rangle)\]

For a suitable $\epsilon > 0$, consider the strategy $\beta'$ that is identical
to $\beta$ everywhere except for those two actions, where:
\begin{align*}
\beta'(I)(a_1) &= \beta(I)(a_1) + \epsilon\\
\beta'(I)(a_2) &= \beta(I)(a_2) - \epsilon
\end{align*}
It follows that:
\allowdisplaybreaks
\begin{align*}
    &\ U^B_i(\beta',\mu | I)\\
    =&\ \sum_{a\in A(I)} \beta'(I)(a) U^B_i(\beta', \mu, | I,a) \\
    =&\ \epsilon\cdot (U^B_i(\beta, \mu, | I,a_1)-U^B_i(\beta, \mu, | I,a_2)) + \sum_{a\in A(I)} \beta(I)(a)
    U^B_i(\beta, \mu, | I,a) \\
    =&\ \epsilon\cdot (U^B_i(\beta, \mu, | I,a_1)-U^B_i(\beta, \mu, | I,a_2)) + U^B_i(\beta,\mu | I) \\
    >&\ U^B_i(\beta,\mu | I).
\end{align*}
\allowdisplaybreaks[0]
This violates local sequential rationality of $\beta$.
Therefore, local sequential rationality implies Property (\ref{eq:sup1}).

Similarly, assume that there exists an action $a^* \in A(I)$ with
$\beta(I)(a^*) = 0$ and $U^B_i(\beta, \mu, | I,a^*) > U^B_i(\beta, \mu, | I)$
such that (\ref{eq:sup2}) does not hold.
Consider the strategy $\beta'$ that is identical to $\beta$ except at I, where
it only plays $a^*$ at $I$:
\begin{align*}
\beta'(I)(a^*) &= 1\\
\beta'(I)(a) &= 0 \qquad \forall a \neq a^*
\end{align*}
It follows that
\begin{align*}
U^B_i(\beta',\mu | I) &= \sum_{a\in A(I)} \beta'(I)(a) U^B_i(\beta', \mu, | I,a) \\
&= U^B_i(\beta, \mu, | I,a^*) > U^B_i(\beta,\mu | I).
\end{align*}
This violates local sequential rationality of $\beta$.
Therefore, local sequential rationality also implies Property (\ref{eq:sup2}).
\hfill$\square$

``$\Leftarrow$''
Let $(\beta, \mu)$ be an assessment for which properties (\ref{eq:sup1}) and
(\ref{eq:sup2}) hold, and let $\beta'$ be a strategy where $\beta'(I') =
\beta(I')$ for all $I'\in \mathcal{I}$ such that $I' \neq I$. Furthermore, let
$A_0 = \{a\in A(I) \mid \beta(I)(a) = 0\}$ and $A_+ = A(I) \setminus A_0$. It
follows that
\begin{align*}
U^B_i(\beta',\mu | I)
& = \sum_{a\in A(I)} \beta'(I)(a) U^B_i(\beta', \mu, | I,a)\\
& = \sum_{a\in A(I)} \beta'(I)(a) U^B_i(\beta, \mu, | I,a)\\
& = \sum_{a\in A_0} \beta'(I)(a) U^B_i(\beta, \mu, | I,a) + \sum_{a\in A_+} \beta'(I)(a) U^B_i(\beta, \mu,
| I,a)\\
& = \sum_{a\in A_0} \beta'(I)(a) U^B_i(\beta, \mu, | I,a) + \sum_{a\in A_+} \beta'(I)(a) U^B_i(\beta, \mu,
| I)\\
& \leq \sum_{a\in A_0} \beta'(I)(a) U^B_i(\beta, \mu, | I) + \sum_{a\in A_+} \beta'(I)(a) U^B_i(\beta, \mu,
| I)\\
& = \sum_{a\in A(I)} \beta'(I)(a) U^B_i(\beta, \mu, | I)\\
& = U^B_i(\beta,\mu | I).
\end{align*}
Note that we first take advantage of the fact that $\beta$ and $\beta'$ are
equivalent outside of $I$. We then use the equations (\ref{eq:sup1}) and
(\ref{eq:sup2}) to replace $U^B_i(\beta, \mu, | I,a)$ by $U^B_i(\beta, \mu, |
I)$ for actions in $A_+$ and $A_0$, where (\ref{eq:sup2}) introduces an
inequality. In the last step, $\beta'(I)(a)$ adds up to 1.

Since $U^B_i(\beta',\mu | I) \leq U^B_i(\beta,\mu | I)$ for any $\beta'$ that only deviates locally, $(\beta, \mu)$
is locally sequentially rational.
\end{proof}

We can finally rewrite equations (\ref{eq:sup1}) and (\ref{eq:sup2}) from
Proposition~\ref{prop:support} as a system of polynomial equations and
inequalities without case distinctions. For all $I \in \mathcal{I}$, $a \in
A(I)$, and $i = N(I)$ we obtain
\begin{align*}
\left(\sum_{h \in I} \mu(I)(h) U^E_i(\beta | \langle h, a \rangle)\right) - U^B_i(\beta,\mu | I)\phantom{\Bigg)} &\leq 0, \text{and}\\
\beta(I)(a) \cdot  \left( \left( \sum_{h\in I}\mu(I)(h) U^E_i(\beta | \langle h, a\rangle)\right)- U^B_i(\beta,
    \mu | I) \right) &= 0.
\end{align*}

An assessment is locally sequentially rational if and only if it satisfies this
system of equations. By Theorem~\ref{thrm:oneshot}, the sequential equilibria of
a game are exactly the consistent assessments that are locally sequentially
rational.
In the following, we will show how consistency can be similarly characterized as
a system of polynomial equations, following the results of
\citeay{kohlberg1997independence}.

\subsection{Consistency}

While sequential rationality enforces that strategies are optimal given players'
beliefs, consistency enforces that beliefs correctly reflect the conditional
probabilities of each history being reached, given players' strategies.
In particular, for information sets that are reached with probability
$P_\beta(I) > 0$, we have
\[\mu(I)(h) = \frac{P_\beta(h)}{P_\beta(I)}= P_\beta(h| I).\]

Note that consistency depends only on the structure of the game tree. In
particular, whether a given assessment is consistent does not depend on the
utilities of the game, nor does it depend on the acting player at each
information set. In the case of $P_\beta(I) = 0$, the restrictions imposed by
consistency can become more complex.

Figure~\ref{fig:consistency} shows the structures of five different game trees,
together with possible consistent assessments. In the simplest cases, the
beliefs correspond directly to the action probabilities (a), or to the
conditional probabilities of the strategies leading to each history (b).
Sometimes, if an information set is not reached, the assessment is consistent
for any belief (c). However, such arbitrary beliefs may further constrain the
beliefs at the next information set (d), and even at different parts of the game
tree (e). In the last two examples, to satisfy consistency, the beliefs at both
information sets must be identical, since they would have to be identical in any
fully mixed assessment (i.e., with only positive action probabilities) that
converges to $(\beta, \mu)$. For a more detailed discussion of how beliefs can
be constrained by consistency, see the paper by \citeay{pimienta2014bayesian}.

We will now follow the work of \citeay{kohlberg1997independence} to represent
consistency by a finite set of polynomial equations. The reduction consists of
several steps: An assessment is consistent if and only if a special system of
linear equations $Ax = b$, where $A$ depends on the structure of the game tree
and $b$ depends on the assessment, has a positive approximate solution. Such a
solution exists if and only if a certain property holds for all vectors $p$ with
$pA = 0$. This property can be written as an equation in $\beta$ and $\mu$.
Finally, the set of relevant vectors $p$ can be reduced so that the system of
equations becomes finite and polynomial without changing the solution set. To
obtain this subset of relevant vectors, we have to compute the extreme
directions of a set of polyhedral cones.

Our contribution is to provide more details on the approach by
\citeauthor{kohlberg1997independence}. We explicitly construct the
linear system $Ax=b$ (Theorem~\ref{thrm:linsys}) and formalize and prove all the
necessary intermediate steps to derive the coefficients and exponents of the
polynomial equations (Propositions~\ref{prop:propK}-\ref{prop:generaltest}). In
\citeauthor{kohlberg1997independence}'s paper, the underlying ideas are stated
informally and without proof.

\subsubsection{Positive Approximate Solutions}
The concept of positive approximate solutions is the first step in our process
of representing consistency as a system of polynomial equations. These are
solutions to systems of equations $Ax=b$, where the vector $b$ may contain
values such as $\infty$ or $-\infty$, or ill-defined expressions such as
$\infty - \infty$ or $\frac{0}{0}$, using the conventions of the
\textit{extended real number line} $\overline{\mathbb{R}}$ (see, e.g.,
\cite{aliprantis19901}). A positive approximate solution is then a series $x^n$
such that each component of each vector $x^n$ is positive and $Ax^n$ converges
to $b$ for each component of $b$ that is well-defined.

\begin{definition}[Positive Approximate Solution of a Linear
System]\label{def:pas} Let $Ax = b$ be a linear system where each $b_i \in
\overline{\mathbb{R}} = \mathbb{R} \cup \{-\infty, \infty\}$ or ill-defined. A
positive approximate solution to this system is a series $(x^n)^{n\in
\mathbb{N}}$ where $x^n > 0$ for all $n \in \mathbb{N}$, and $\lim_{n \to
\infty} (Ax^n)_i = b_i$ for all $i$ such that $b_i$ is well-defined.
\end{definition}

\def\nps{N_{\textit{pairs}}}
\def\nas{N_{\textit{actions}}}

We will now construct a linear system from an assessment $(\beta, \mu)$ in such
a way that it has a positive approximate solution if and only if the assessment
is consistent.

\begin{theoremrep}[Linear System for consistency]\label{thrm:linsys}
Let $\{a_i \mid i\in \{1, \ldots, \nas\}\}$ be the set of all actions in an
extensive-form game (we assume without loss of generality that each action can
be played in exactly one information set), and let $\{(h_i^1, h_i^2) \mid i \in
\{1, \ldots, \nps\}\}$ be the set of all history pairs such that and both
histories $h_i^1$ and $h_i^2$ are in the same information set. Furthermore, let
$M = \nas + \nps$. Then an assessment $(\beta, \mu)$ is consistent if and only
if the linear system $Ax = b$ which is defined as follows has a positive
approximate solution.
\allowdisplaybreaks
\begin{align*}
A &= \bigg[
    \begin{array}{c}
        \tilde{A} \\
        I_{\nas}
    \end{array} \bigg], \text{ where }\tilde{A} \in \mathbb{R}^{\nps \times \nas} \text{ is defined as below:}\\
\tilde{A}_{i,j} &=
\begin{cases}
~ 1  & \textrm{if } a_j \in h_i^1 \text{ and } a_j \notin h_i^2 \\
~ -1 & \textrm{if } a_j \notin h_i^1 \text{ and } a_j \in h_i^2 \\
~ 0  & \textrm{otherwise}
\end{cases}\\
b_i &= log(\alpha_i) - log(\gamma_i) \text{ for } i \in \{1, \ldots, M\} \text{, where}\\
\alpha_i &= 
\begin{cases}
~ \mu(h_i^1) & \textrm{if } i\in \{1, \ldots, N_{pairs}\} \\
~ \beta(a_{i - N_{pairs}}) & \textrm{otherwise}
\end{cases}\\
\gamma_i &=
\begin{cases}
~ \mu(h_i^2) & \textrm{if } i\in \{1, \ldots, N_{pairs}\} \\
~ 1 & \textrm{otherwise}
\end{cases}
\end{align*}
\allowdisplaybreaks[0]
\end{theoremrep}

Note that $A \in \mathbb{R}^{M \times \nas}$, $b\in \overline{\mathbb{R}}^M$,
and $x\in \mathbb{R}^{\nas}$.

\begin{proofsketch}
Both positive approximate solutions and consistency depend on the existence of a
convergent series. For consistency, we need a series of fully-mixed assessments
$(\beta^n, \mu^n)$, where the $\mu^n$ are determined by the $\beta^n$ via Bayes'
rule. We
define a multiplicative system for which the positive approximate solutions
(which are defined similarly as for linear systems) are exactly these series.
For $I\in \mathcal{I}$, each pair $h_1, h_2\in I$ and action $a_i\in I$, we have:
\begin{align}
    \frac{\prod_{a_i\in h_1} x_i}{\prod_{a_i\in h_2} x_i} &
  = \frac{\mu(I)(h_1)}{\mu(I)(h_2)}\\
x_i &= \beta(I)(a_i)
\end{align}
The linear system is then obtained by taking the logarithm of this system. It
has a positive approximate solution if and only if the multiplicative system has
one. The full proof is in the appendix.
\end{proofsketch}

\begin{proof}
Given an assessment $(\beta, \mu)$, consider first the following system of
multiplicative equations:
\begin{align*}
\frac{\prod_{a_i\in h_1} x_i}{\prod_{a_i\in h_2} x_i} &= \frac{\mu(I)(h_1)}{\mu(I)(h_2)} & (M1)\\
x_i &= \beta(I)(a_i) & (M2)
\end{align*}

For each $I\in \mathcal{I}$, each pair $h_1, h_2\in I$ and action $a_i\in I$.
Such a system of the form
$$\alpha_m, \gamma_m \in [0, \infty) ,\ \ \  I_m, J_m \subset \{1, 2, \ldots, N\}$$
$$\frac{\prod_{i\in I_m}x_i}{\prod_{i\in J_m}x_j} = \frac{\alpha_m}{\gamma_m}$$
$$m\in \{1, 2, \ldots, M\}$$
is said to have a positive approximate solution $(x^n)_{n\in \mathbb{N}}$ if
$$x^n > 0,\ \forall n\in \mathbb{N},$$
$$\lim_{n\to\infty} \frac{\prod_{i\in I_m}x^n_i}{\prod_{i\in J_m}x^n_j} = \frac{\alpha_m}{\gamma_m} ,
\forall m:\frac{\alpha_m}{\gamma_m} \neq \frac{0}{0}.$$

We can transform such a multiplicative system to a linear system by taking the
logarithm:
\begin{align*}
&\ \frac{\prod_{i\in I_m}x_i}{\prod_{i\in J_m}x_j} = \frac{\alpha_m}{\gamma_m} \\
\iff &\ log(\frac{\prod_{i\in I_m}x_i}{\prod_{i\in J_m}x_j}) = log(\frac{\alpha_m}{\gamma_m}) \\
\iff &\ \sum_{i\in I_m}log(x_i) - \sum_{i\in J_m}log(x_j) = log(\alpha_m) - log(\gamma_m) \\
\iff &\ A\tilde{x} = b,\ A\in \mathbb{R}^{M \times N},\ b\in \overline{\mathbb{R}}^M, \tilde{x}\in \mathbb{R}^N
\end{align*}
Here, for $i\in \{1, \ldots, M\}, j\in \{1, \ldots, N\}$, we have
$$b_i = log(\alpha_i) - log(\gamma_i),\ \tilde{x}_j = log(x_j),$$
$$A_{i,j} = \left\{
\begin{array}{lll}
    1  & \textrm{ if } j\in I_i, j\not\in J_i \\
    -1 & \textrm{ if } j\not\in I_i, j\in J_i \\
    0   & \, \textrm{otherwise}.
\end{array}
\right. $$

The multiplicative system then has a positive approximate solution if and only
if the multiplicative system has a solution.

Applying this transformation to the multiplicative system above yields the
linear system described in $\ref{thrm:linsys}$. Note that the index sets $I_i,
J_i$ for  correspond to the actions on the path of $h^1_i, h^2_i$ for equations $(M1)$, resulting in $\tilde{A}$, and identify $x_i$ with $\beta(I)(a_i)$
for equations $(M2)$, with $I_{i+N_{pairs}} = \{i\}$ and
$J_{i+N_{pairs}} = \emptyset$. It remains to show that the multiplicative system
has a solution if and only if the assessment is consistent.

$"\Leftarrow"$ Assume ($\beta, \mu$) is consistent, then there exists $(\beta^n,
\mu^n)$ with
\begin{align*}
\beta^n(I)(a) &> 0 &\forall n\in \mathbb{N},\ \forall I\in \mathcal{I},\ a\in A(I)\\
\mu_n(I)(h) &= \frac{P_{\beta^n}(h)}{P_{\beta^n}(I)} &\forall n\in \mathbb{N},\ \forall I\in \mathcal{I},\ h\in I
\end{align*}
\begin{align*}
\lim_{n\to\infty}(\beta^n, \mu^n) &= (\beta, \mu)
\end{align*}

Then $x^n_i = \beta^n(I)(a_i)$ is a positive approximate solution to the system of
equations above. We see that positivity and $(M2)$ follow directly from definition.
\begin{align*}
    x^n_i &= \beta^n(I)(a_i) > 0 \\
    \lim_{n\to\infty}x^n_i &= \lim_{n\to\infty}\beta^n(I)(a_i) = \beta(I)(a_i)
\end{align*}
For $(M1)$ we note that $\forall h_1, h_2\in I$,
\begin{flalign*}
    \frac{\prod_{a_i\in h_1} x^n_i}{\prod_{a_i\in h_2} x^n_i}
    & = \frac{\prod_{a_i\in h_1} \beta^n(I)(a_i)}{\prod_{a_i\in h_2} \beta^n(I)(a_i)} &\\
    &
    = \frac{P_{\beta^n}(h_1)}{P_{\beta^n}(h_2)} = \frac{P_{\beta^n}(h_1)}{P_{\beta^n}(h_2)} \cdot \frac{P_{\beta^n}(I)}{P_{\beta^n}(I)} = \frac{\mu^n(I)
        (h_1)}{\mu^n(I)(h_2)}.&
\end{flalign*}
This gives us, $\forall h_1, h_2$ where $\frac{\mu(I)(h_1)}{\mu(I)(h_2)} \neq
\frac{0}{0}$:

$$\lim_{n\to\infty} \frac{\prod_{a_i\in h_1} x^n_i}{\prod_{a_i\in h_2} x^n_i} =
\lim_{n\to\infty} \frac{\mu^n(I) (h_1)}{\mu^n(I)(h_2)} =
\frac{\mu(I)(h_1)}{\mu(I)(h_2)}$$

"$\Rightarrow$" Assume the system of equations has a positive approximate
solution $x^n$. Then construct an assessment $(\beta^n, \mu^n)$ as:
$$\beta^n(I)(a_i) = \frac{x_i^n}{\sum_{a_j\in A(I)}x_j^n},\ \ \ \mu^n(I)(h) =
\frac{P_{\beta^n}(h)}{P_{\beta^n}(I)} $$

We can see that $\lim_{n\to\infty}\beta^n(I)(a_i) = \lim_{n\to\infty}x^n_i =
\beta(I)(a_i)$ and that  $\beta^n$ is fully mixed.
It remains to show that $\lim_{n\to\infty} \mu^n(I)(h) = \mu(I)(h)$. Consider
that for any pair $h_1, h_2\in I$ where $\frac{\mu(I)(h_1)}{\mu(I)(h_2)} \neq
\frac{0}{0}$:
\begin{flalign*}
\frac{\mu^n(I)(h_1)}{\mu^n(I)(h_2)}
    = \frac{P_{\beta^n}(h_1)}{P_{\beta^n}(h_2)}
    = \frac{\prod_{a_i\in h_1} \beta^n(I)(a_i)}{\prod_{a_i\in h_2} \beta^n(I)(a_i)}
    = \frac{\prod_{a_i\in h_1} x^n_i}{\prod_{a_i\in h_2} x^n_i} &
\end{flalign*}
If $\mu(I)(h) > 0$, we get that
\begin{flalign*}
    \lim_{n\to\infty} \frac{1}{\mu^n(I)(h)} & = \lim_{n\to\infty} \frac{\sum_{h'\in I}\mu^n(I)(h')}{\mu^n(I)(h)} & \\
    & = \sum_{h'\in I} \lim_{n\to\infty} \frac{\mu^n(I)(h')}{\mu^n(I)(h)} & \\
    & = \sum_{h'\in I} \frac{\mu(I)(h')}{\mu(I)(h)}  = \frac{1}{\mu(I)(h)}. &
\end{flalign*}
Therefore $\lim_{n\to\infty} \mu^n(I)(h) = \mu(I)(h)$.
Otherwise, if $\mu(I)(h) = 0$, let $h'\in I$ be another history such that
$\mu(I)(h') > 0$. Because of the previous result, we know that
$\lim_{n\to\infty}\mu^n(I)(h') = \mu(I)(h')$ and it follows that
\begin{flalign*}
    \frac{\lim_{n\to\infty} \mu^n(I)(h)}{\mu(I)(h')}
    & = \frac{\lim_{n\to\infty} \mu^n(I)(h)}{\lim_{n\to\infty} \mu^n(I)(h')} & \\
    & = \lim_{n\to\infty} \frac{\mu^n(I)(h)}{\mu^n(I)(h')} & \\
    & = \frac{\mu(I)(h)}{\mu(I)(h')}.
\end{flalign*}
Multiplying by $\mu(I)(h') \neq 0$ gives us that $\lim_{n\to\infty}\mu^n(I)(h) =
\mu(I)(h)$. Since $\lim_{n\to\infty}\mu^n(I)(h) = \mu(I)(h)$ for all $h\in I$,
the assessment is consistent.
\end{proof}

Note that some of the $b_i$ can be ill-defined. This is the case if $\alpha_i =
\gamma_i = 0$ and thus $b_i = \log(0) - \log(0) = \infty - \infty$. Since the
existence of positive approximate solutions only depends on the equations where
$b_i$ is well-defined, we can reduce the system such that the equations where
$b_i$ is ill-defined are omitted. From now on, we will assume that all $b_i$ are
well-defined.

\subsubsection{Existence of a Positive Approximate Solution}
\citeay{kohlberg1997independence} give a result for the existence of a positive
approximate solution to a linear system. We restate this here without proof:

\begin{theorem}[Solution Existence for Linear Systems, \cite{kohlberg1997independence}]
\label{thm:solex}
A linear system $Ax = b$ has a positive approximate solution if and only if the
following property holds for all $p \in \mathbb{R}^M$ where $pA = 0$:
\begin{align}\label{eq:propK}
\sum_{p_i \neq 0} p_i b_i = 0 \textrm{ or }
\sum_{p_i \neq 0} p_i b_i  \textrm{ is ill-defined}
\end{align}
\end{theorem}
The sum can be ill-defined if it contains the expressions $0 \cdot \infty$ or
$\infty - \infty$. Since we only sum over $p_i \neq 0$, we only need to consider
the second case.
Importantly, we can write Property (\ref{eq:propK}) for a given vector
$p \in \mathbb{Z}^M$ as a polynomial equation.

\begin{proposition}\label{prop:propK}
Consider a linear system $Ax=b$ where $b_i = log(\alpha_i) - log(\gamma_i)$ for
some $\alpha_i,\gamma_i\in \mathbb{R}$. Then Property (\ref{eq:propK}) holds for
some vector $p \in \mathbb{Z}^M$ if and only if the following equation is
satisfied:
\begin{align}\label{eq:eqK}
  \prod_{p_i > 0} \alpha_i^{p_i} \prod_{p_i < 0} \gamma_i^{-p_i}
= \prod_{p_i > 0} \gamma_i^{p_i} \prod_{p_i < 0} \alpha_i^{-p_i}
\end{align}
\end{proposition}

\begin{proof}
Consider first the case where $\sum_{p_i \neq 0} p_i b_i$ is ill-defined. Here,
we know that Property (\ref{eq:propK}) always holds. Therefore, we only need to
show that equation (\ref{eq:eqK}) is satisfied.
For $\sum_{p_i \neq 0} p_i b_i$ to be ill-defined, there must exist indices $i,
j$ such that $p_i b_i = \infty$ and $p_j b_j = -\infty$. Here the sum $p_i b_i +
p_j b_j = \infty - \infty$ is ill-defined. This happens if either $p_i > 0$ and
$\gamma_i =0$, or $p_i < 0$ and $\alpha_i = 0$. In any case, $\prod_{p_i >0}
\gamma_i^{p_i}\prod_{p_i < 0} \alpha_i^{-p_i} = 0$. Similarly, we know that
either $p_j > 0$ and $\alpha_i = 0$, or $p_j < 0$ and $\gamma_j = 0$, which
means that $\prod_{p_i > 0} \alpha_i^{p_i} \prod_{p_i < 0} \gamma_i^{-p_i} = 0.$
Therefore, equation (\ref{eq:eqK}) is satisfied.

In the case where $\sum_{p_i \neq 0} p_i b_i$ is well-defined, Property
(\ref{eq:propK}) holds for $p$ if and only if $\sum_{p_i \neq 0} p_i b_i = 0$.
We obtain equation (\ref{eq:eqK}) by taking the exponential function and then
multiplying by all the terms with a negative exponent.
\begin{align*}
\sum_{p_i \neq 0} p_i b_i = 0
& \iff \sum_{p_i \neq 0} p_i (log(\alpha_i) - log(\gamma_i)) = 0 & \\
& \iff \prod_{p_i \neq 0}\bigg(\frac{\alpha_i}{\gamma_i}\bigg)^{p_i} = 1 & \\
& \iff \prod_{p_i > 0} \alpha_i^{p_i} \prod_{p_i < 0} \gamma_i^{-p_i} = \prod_{p_i > 0} \gamma_i^{p_i}
\prod_{p_i < 0}
    \alpha_i^{-p_i}
\end{align*}
Note that some of the terms we multiply by can be equal to zero. If this is the
case, all of the terms with positive exponents are nonzero, since otherwise
$\sum_{p_i \neq 0} p_i b_i $ would be ill-defined. Here, neither equation is
satisfied and their equivalency still holds.
\end{proof}

\subsubsection{A Finite System of Equations}
To write consistency as a finite system of polynomial equations, we have to
solve two problems: In Theorem~\ref{thm:solex}, we consider vectors which can
have non-integer components. This means that we cannot use
Proposition~\ref{prop:propK} to obtain an equivalent polynomial equation.
Furthermore, there are infinitely many vectors $p$ with $pA=0$ (except in
perfect information games where $A$ is the identity matrix $I_{\nas}$).

We now reduce the set of relevant vectors to a finite one.

Let $W^A_b \subseteq \{p \mid pA=0\}$ be the set of all $p$ such that $\sum_{p_i
\neq 0} p_i b_i$ is well-defined. We then only need to check
Property~(\ref{eq:propK}) for all $p \in W^A_b$, since we already know that it
holds for all $p \notin W^A_b$.

Consider again that $\sum_{p_i \neq 0} p_i b_i$ is well-defined if there are no
indices $i$ and $j$ such that $p_i b_i = \infty$ and $p_j b_j = -\infty$, For
any $p \in W^A_b$, either all infinite terms of the sum must be positive, or all
infinite terms must be negative. We can thus write $W^A_b = C^A_b \cup -C^A_b$
where
\begin{align*}
C^A_b &= \{p \mid pA = 0 \land p_i \geq 0 \textrm{ if } b_i = \infty \land
p_i \leq 0 \textrm{ if } b_i = -\infty\}, \text{ and}\\
-C^A_b &= \{p \mid pA = 0 \land p_i \leq 0 \textrm{ if } b_i = \infty \land
p_i \geq 0 \textrm{ if } b_i = -\infty\}.
\end{align*}

We now show that it is sufficient to check Property (\ref{eq:propK}) for all $p
\in C^A_b$. As we will see, we do not need to consider $p \in -C^A_b$.

\begin{proposition}\label{prop:cone} Let $Ax = b$ be the linear system from
Theorem~\ref{thrm:linsys}. Then the assessment $(\beta, \mu)$ is consistent if
and only if Property (\ref{eq:propK}) is satisfied for all $p\in C^A_b$.
\end{proposition}

\begin{proof}
By Theorem~\ref{thm:solex}, the assessment is consistent if and only if Property
(\ref{eq:propK}) is satisfied for all $p$ where $pA = 0$. Since the property is
satisfied if $\sum_{p_i \neq 0} p_i b_i$ is ill-defined, we do not need to
consider vectors $p \not\in W^A_b$. For $p \in W^A_b$, note that $\sum_{p_i \neq
0} p_i b_i = 0 \iff \sum_{p_i \neq 0} -p_i b_i = 0$. This means that Property
(\ref{eq:propK}) holds for $p$ if and only if it holds for $-p$. The assessment
is thus consistent if and only if Property (\ref{eq:propK}) is satisfied for all
$p\in C^A_b$.
\end{proof}

As we can see, $C^A_b$ is an intersection of half spaces and thus a pointed
polyhedral cone:
\begin{align*}
C^A_b = \{p \mid pA= 0\}
\cap \bigcap_{b_i = \infty} \{p \mid p_i \geq 0\}
\cap \bigcap_{b_i = -\infty} \{p \mid p_i \leq 0\}
\end{align*}

We can alternatively represent $C^A_b$ as the set of all conical combinations of
finitely many vectors $\{e_1, \ldots,  e_k\}$ such that
\[C^A_b = \{\lambda_1 e_1 + \ldots + \lambda_k e_k \mid \lambda_k\in \mathbb{R}^+\}.\]

These vectors are called extreme directions (or conical basis) of $C^A_b$.
Transforming one representation into the other can be done with the double
description method \cite{zolotykh2012new}, which we will discuss in
Section~\ref{sec:implementation}.
Note that in the cases where all $b_i$ are infinite, the extreme directions of
$C^A_b$ are unique modulo scaling. Otherwise, this is not necessarily the case.
Furthermore, because the entries of $A$ are always integers, each of the extreme
directions can be scaled to have integer components. This is another result by
\citeay{kohlberg1997independence}. We chose an arbitrary conical basis
$ED(C^A_b)$ which has this property. This will allow us to use
Proposition~\ref{prop:propK} to obtain a system of polynomial equations.
We show that if Property (\ref{eq:propK}) holds for two vectors of a cone, it
also holds for arbitrary conical combinations. This allows us to reduce the
system to a finite one.

\begin{proposition}\label{prop:conecomb}
If Property (\ref{eq:propK}) holds for two vectors $x, y \in C^A_b$ then it must
also hold for any conical combination $z = \alpha x + \beta y$, $\forall \alpha,
\beta \in \mathbb{R}^+$.
\end{proposition}
\begin{proof}
Since $C^A_b$ is a cone, it follows that $z\in C^A_b$. Therefore $\sum_{z_i \neq
0} z_i b_i$ is well defined and Property (\ref{eq:propK}) holds if $\sum_{z_i
\neq 0}z_i b_i = 0$. We have $\sum_{x_i \neq 0} x_i b_i = \sum_{y_i \neq 0} y_i
b_i = 0$. Let $z = \alpha x + \beta y$, $\alpha, \beta \in \mathbb{R}^+$. We
split $z_i\neq 0$ into the cases $(x_i=0, y_i\neq 0)$, $(x_i\neq 0, y_i=0)$, and
$(x_i\neq 0, y_i\neq 0)$.
\allowdisplaybreaks
\begin{align*}
    &\hspace{-12pt}\sum_{(\alpha x + \beta y)_i \neq 0} (\alpha x + \beta y)_i\cdot b_i\\
    = &\sum_{\substack{x_i \neq 0\\ y_i = 0}} (\alpha x + \beta y)_i\cdot b_i
    + \sum_{\substack{x_i = 0\\ y_i \neq 0}} (\alpha x + \beta y)_i\cdot b_i
    + \sum_{\substack{x_i \neq 0\\ y_i \neq 0}} (\alpha x + \beta y)_i\cdot b_i \\
    = &\sum_{\substack{x_i \neq 0\\ y_i = 0}} \alpha x_i b_i
    + \sum_{\substack{x_i = 0\\ y_i \neq 0}} \beta y_i b_i
    + \sum_{\substack{x_i \neq 0\\ y_i \neq 0}} \alpha x_i b_i + \beta y_i b_i \\
    = &\sum_{\substack{x_i \neq 0\\ y_i = 0}} \alpha x_i b_i + \sum_{\substack{x_i \neq 0\\ y_i \neq 0}} \alpha x_i b_i
    + \sum_{\substack{x_i = 0\\ y_i \neq 0}} \beta y_i b_i + \sum_{\substack{x_i \neq 0\\ y_i \neq 0}} \beta y_i b_i\\
    = &\sum_{x_i \neq 0} \alpha x_i b_i + \sum_{y_i \neq 0} \beta y_i b_i 
    = 0
\end{align*}
\allowdisplaybreaks[0]

In the first step of our transformation, there might be some $j$ where
$x_j\neq 0$, $y_j\neq 0$ but $(\alpha x + \beta y)_j = z_j = 0$. The terms $z_j
b_j$ would normally not be included in $\sum_{z_i\neq 0} z_i b_i$. In those
cases, it follows that $\alpha x_j = -\beta y_j$ where $\alpha,\beta > 0$,
therefore $x_j$ and $y_j$ have a different sign. Since $x, y\in C^A_b$, $x_i$
and $y_i$ must have the same sign whenever $b_i$ is infinite. Thus $b_j$ is
finite and $z_j b_j = 0$. We can therefore add these terms to $\sum_{z_i \neq 0}
z_i b_i$ while preserving equality.%
\footnote{The same argument does not work for linear combinations.
Assuming $x_j \neq 0$, $y_j \neq 0$, and $(\alpha x + \beta y)_j = z_j = 0$, it
is possible that $x_j$ and $y_j$ have the same sign, since $\alpha$ and $\beta$
can be negative. Then, $b_j$ can be infinite, in which case $z_j b_j = 0 \cdot
\infty$ is ill-defined.}
\end{proof}

We can now formalize a finite test for consistency.

\begin{proposition}[Finite Consistency Test]\label{prop:finitetest} Let $Ax = b$
be the linear system from Theorem~\ref{thrm:linsys}. Then the assessment
$(\beta, \mu)$ is consistent if and only if Property (\ref{eq:propK}) holds for
all $p\in ED(C^A_b)$.
\end{proposition}
\begin{proof}
If the assessment is consistent, then Property (\ref{eq:propK}) must hold for
all $p\in \{p\mid pA=0\}$ due to Theorem (\ref{prop:propK}), so it also holds
for all $p\in ED(C^A_b) \subseteq C^A_b \subseteq \{p\mid pA=0\}$. If Property
(\ref{eq:propK}) holds for all $p\in ED(C^A_b)$, then it holds for all $p\in
C^A_b$ because of Proposition \ref{prop:conecomb} and because each $p$ can be
written as conical combination of $ED(C^A_b)$. The assessment is then consistent
due to Proposition \ref{prop:cone}.

Proposition  \ref{prop:conecomb} also implies that the choice of $ED(C^A_b)$ is
irrelevant, since if Property (\ref{eq:propK}) holds for one set of extreme
directions, then it holds for the whole cone and thus for any other set of
extreme directions.
\end{proof}

\subsubsection{Finding all Consistent Assessments}
We now have a finite test for proving consistency of a given assessment $(\beta,
\mu)$. However, we still cannot easily describe the set of all consistent
assessments. This is because the test from Proposition~\ref{prop:finitetest}
depends on the specific cone $C^A_b$, which depends on the right-hand side of
the linear system $Ax = b$, which depends on the exact values of $(\beta, \mu)$.
More precisely, it is the actions with $\beta(I)(a) = 0$ and the beliefs with
$\mu(I)(h) = 0$ that determine which $b_i$ are finite, $\infty$, or $-\infty$.
Assuming that $A$ (which only depends on the game tree) is fixed, only the
positions of infinite values in $b$ are relevant for $C^A_b$. Formally, if $b_i'
= \infty \iff b_i = \infty$ and $b_i' = -\infty \iff b_i = -\infty$, then
$C^A_{b'} = C^A_b$.

Since we want to characterize all sequential equilibria of a game, we need to
find a criterion that works for arbitrary values of $(\beta, \mu)$. As we will
see, we can use the extreme directions of all cones $\mathcal{C}^A = \{C^A_b
\mid b_i \in \{-\infty, 0, \infty\}, \forall i\}$ relevant to $A$. The set of
extreme directions of all cones relevant to $A$ is defined as follows:
\[\mathcal{E}^A = \bigcup_{C \in \mathcal{C}^A} ED(C) =
\bigcup_{b_i \in \{-\infty, 0, \infty\}, \forall i}ED(C^A_b).\]

We now show how we can use $\mathcal{E}^A$ to characterize consistency
independently of the exact values of $(\beta, \mu)$.

\begin{proposition}[General Consistency Test]\label{prop:generaltest}
Let $(\beta, \mu)$ be an assessment and $Ax=b$ be the linear system from
Theorem~\ref{thrm:linsys}. Then $(\beta, \mu)$ is consistent if and only if
Property (\ref{eq:propK}) holds for all $p \in \mathcal{E}^A$.
\end{proposition}
\begin{proof}
``$\Leftarrow$'' If Property (\ref{eq:propK}) holds for all $p\in
\mathcal{E}^A$, then it holds specifically for all $p\in ED(C^A_b)\subseteq
\mathcal{E}^A$. Thus, by Proposition~\ref{prop:finitetest}, the assessment is
consistent.
``$\Rightarrow$'' If the assessment is consistent, then by
Theorem~\ref{thm:solex}, Property (\ref{eq:propK}) must hold for any $p$ such
that $pA=0$. Since $\mathcal{E}^A \subseteq \{p \mid pA = 0\}$, Property
(\ref{eq:propK}) holds for all $p \in \mathcal{E}^A$.
\end{proof}

\subsubsection{Polynomial Equations}

We can now express consistency as a finite system of polynomial equations. This
result follows directly from Propositions \ref{prop:propK} and
\ref{prop:generaltest}.

\begin{theorem}[Consistency as Polynomial Equations,
\cite{kohlberg1997independence}]\label{thrm:polyeq}
Let $A$, $\alpha$, and $\gamma$ be defined as in Theorem~\ref{thrm:linsys}. Then
an assessment $(\beta, \mu)$ is consistent if and only if for all $p\in
\mathcal{E}^A$,
\[\prod_{p_i > 0} \alpha_i^{p_i} \prod_{p_i < 0} \gamma_i^{-p_i} = \prod_{p_i >
0} \gamma_i^{p_i} \prod_{p_i < 0} \alpha_i^{-p_i}.\]
\end{theorem}

\section{Implementation}\label{sec:implementation}
In the previous section we have seen how sequential rationality can be expressed
as a system of polynomial equations and inequalities if we assume consistency.
We have also seen how consistency can be expressed as a system of polynomial
equations. Together, these equations characterize the set of all sequential
equilibria.

\subsection{Equations}
First, we recapitulate the entire system of equations and inequalities. The
variables in our equations are the probabilities $\beta(I)(a)$ for each action
$a$ to be played at its information set $I$, and the beliefs $\mu(I)(h)$ that
players assign to each history $h$ at $I$. The equations are quantified over all
$I \in \mathcal{I}$ (with $i = N(I)$), $a \in A(I)$, and $p \in \mathcal{E}^A$:

\twoeq{\phantom{\sum_{h\in I}}\beta(I)(a) &\geq 0\label{eq:prob1}}{\sum_{a\in A(I)}\beta(I)(a) &= 1\label{eq:prob2}}
\twoeq{\phantom{\sum_{h\in I}}\mu(I)(h) &\geq 0\label{eq:bel1}}{\sum_{h\in I}\mu(I)(h) &= 1\label{eq:bel2}}
\begin{align}
\left( \sum_{h\in I}\mu(I)(h) U^E_i(\beta | \langle h, a\rangle) \right)- U^B_i(\beta, \mu | I) \phantom{\Bigg)} &\leq 0
\label{eq:sr1}\\
\beta(I)(a) \cdot \left(\left( \sum_{h\in I}\mu(I)(h) U^E_i(\beta | \langle h, a\rangle) \right)- U^B_i(\beta,
\mu | I)\right) &= 0
\label{eq:sr2}
\end{align}
\begin{align}
\prod_{p_i > 0} \alpha_i^{p_i} \prod_{p_i < 0} \gamma_i^{-p_i} = \prod_{p_i
> 0} \gamma_i^{p_i} \prod_{p_i < 0} \alpha_i^{-p_i}
\label{eq:cons}
\end{align}

Equations (\ref{eq:prob1}-\ref{eq:bel2}) ensure that strategies $\beta(I)$ and
beliefs $\mu(I)$ are probability distributions. Equations (\ref{eq:sr1}) and
(\ref{eq:sr2}) correspond to the sufficient and necessary conditions for local
sequential rationality (Proposition~\ref{prop:support}). The equations of type
(\ref{eq:cons}) ensure consistency.
That is, $\mathcal{E}^A$ is the set of extreme directions of all cones from
\[\mathcal{C}^A = \{C^A_b ~|~ b_i \in \{-\infty, 0, \infty\}, i\in\{1,
\ldots, M\}\}\text{, where}\]
\[C^A_b = \{p | pA=0\}\  \cap \bigcap_{i:\ b_i = \infty} \{p | p_i \geq 0\} \cap
 \bigcap_{i:\ b_i = -\infty} \{p |  p_i \leq 0\},\]
and $A$, $\alpha$, and $\gamma$ are defined as in Theorem~\ref{thrm:linsys}. Next,
we will briefly detail how to compute these extreme directions.

\subsection{Finding all Extreme Directions}\label{sec:cones}

A naive approach is to compute the extreme direction of each cone separately.
This can be done with the so-called \emph{double description method}
\cite{zolotykh2012new}. This algorithm computes the extreme directions of a
given cone by iteratively considering all constraints, calculating new extreme
directions at each iteration based on the current constraint and the previously
computed extreme directions.

For example, to determine the extreme directions of $C^A_{(\infty, \infty,
\infty)}$, the algorithm computes the extreme directions of $C^A_{(0, 0, 0)}$,
$C^A_{(\infty, 0, 0)}$, and $C^A_{(\infty, \infty, 0)}$ as intermediate steps.
As we can see, running the algorithm for each cone separately is inefficient
because the extreme directions of some cones are computed exponentially often as
intermediate steps.
We can avoid this by computing the extreme directions of cones with fewer
constraints first and memorizing the results for the computation of cones with
more constraints. Consider the following collection of sets:
\begin{align*}
\mathcal{C}_i &= \{ C^A_b \mid  b_j = 0, \forall j \geq i \}
              &\forall i \in \{1, \ldots, M+1\}
\end{align*}

Our algorithm first computes the extreme directions of $C^A_{(0, \ldots, 0)}$
(which is the only cone in $\mathcal{C}_1$) and then iteratively computes the
extreme directions for all the cones in the sets $\mathcal{C}_2, \ldots,
\mathcal{C}_{M+1}$. Importantly, each cone in $\mathcal{C}_{i+1}$ corresponds to
a cone in $\mathcal{C}_i$ with at most one constraint added ($p_i \leq 0$ or
$p_i \geq 0$). The computation of new extreme directions for that cone thus
corresponds to performing a single additional step of the double description
method.

The way we iterate over the cones ensures that our algorithm only has to compute
the extreme directions of each cone once. However, each cone may still be
relevant to the set of extreme directions. In general, there are $3^{M}$ cones,
where $M$ is the number of actions plus the number of pairs of histories in the
same information set. For larger games, the number becomes prohibitively large.

The number of cones can be reduced by identifying and removing actions that are
not relevant to consistency. These are the actions such that for all pairs of
histories in the same information set, the action is either on the path of both
histories, or on neither.

We can further optimize our approach by pruning cones for which we can
determine that no additional extreme directions will be introduced. The full
algorithm can be found in the Appendix.

\begin{toappendix}
A single step of the double description method, adapted to our kind of cones and
restrictions, works in the following way:
For a cone $C^A_b$ and a new restriction $p_i \geq 0$, we partition its extreme
directions $ED(C^A_b)$ into three sets:
$$U_+ = \{u \mid u\in ED(C^A_b),\ u_i > 0\}$$
$$U_- = \{u \mid u\in ED(C^A_b),\ u_i < 0\}$$
$$U_0 = \{u \mid u\in ED(C^A_b),\ u_i = 0\}$$

New extreme directions are generated by all pairs of $u\in U_+, v\in U_-$ that
are \textit{adjacent} in $C^A_b$. For general restrictions $\langle a, p \rangle
\leq 0$ they are calculated as $w =  \langle a, u \rangle v - \langle a, v
\rangle u$, but since our restrictions are all of the form $a_ip_i \leq 0$
($a_i\in\{-1, 1\}$), this simplifies to
$$w =  a_i \cdot u_i\cdot  v - a_i\cdot v_i\cdot u = a_i(u_i \cdot v - v_i \cdot
u).$$
for restriction $p_i\geq 0$ ($a_i = -1$). The new extreme directions are then
$$U_{new} = \{w = v_i \cdot u - u_i\cdot v \mid u\in U_+, v\in U_-, (u, v)\
adjacent\textrm{ in }C^A_b\},$$
and the new cone $C^A_{b'} = C^A_b \cap \{p\mid p_i > 0\}$ has extreme
directions
$$ED(C^A_b) = U_+ \cup U_0 \cup U_{new}.$$

For the opposite restriction $p_i \leq 0$ we get the same partition of extreme
directions, only that $U_+$ and $U_-$ are swapped. We can see that a pair $(v,
u)$ produces the same extreme direction $w'$ for $p_i \leq 0$ as $(u, v)$ does
for $p_i\geq 0$:
$$w' = 1 \cdot (v_i \cdot u - u_i \cdot v) = -1 \cdot (u_i \cdot v - v_i \cdot
u) = w$$
This means we can compute the extreme directions of the two cones
$$C^A_{b'} = C^A_b \cap \{p\mid p_i \geq 0\} \textrm{ and } C^A_{b''} = C^A_b
\cap \{p\mid p_i \leq 0\}$$
at the same time.

We iterate over all cones using the collection of sets described above:
\[\mathcal{C}_i = \{ C^A_b \mid \forall j \geq i : b_j = 0 \} \forall  \in
\{1, \ldots, M\}, \mathcal{C}_{M+1}=\mathcal{C} \]
At iteration $i$, we consider all cones in the set $\mathcal{C}_{i}$. For each
$C^A_b\in \mathcal{C}_{i}$, we perform a step of the double description method
as described above to calculate the extreme directions of the two cones:
$$C^A_{b'} \textrm{ where } b' = (b_1, \ldots, b_{i-1}, \infty, 0, \ldots, 0)$$
$$C^A_{b''} \textrm{ where } b'' = (b_1, \ldots, b_{i-1}, -\infty, 0, \ldots, 0)$$
By doing this for all $C^A_b\in \mathcal{C}_{i}$, this iteration has calculated
the extreme directions of all cones in $\mathcal{C}_{i+1}$ based on the extreme
directions of all cones in $\mathcal{C}_{i}$. Note that in practice, we
encode $b$ with values in $\{-1, 0, 1\}$, with $1$ representing $\infty$ and
$-1$ representing $-\infty$.

We start with the extreme directions of $C^A_{(0, \ldots, 0)} = \{p |  pA = 0\}$. Since it
is a full vector space, any basis 
$v^1, \ldots, v^n$ of $\{p |  pA = 0\}$ together with its negatives $-v^1, \ldots, -v^n$
is a set of extreme directions of $C^A_{(0, \ldots, 0)}$. We select a basis that has only integer components
, which is possible since $A$ only has integer components.

\begin{algorithm}[ht]
    \caption{Modified Double Description Method}
    \begin{algorithmic}
        \State $b_0 \gets (0, 0, \ldots, 0)$
        \State $b_0\_ed \gets \{v^1, \ldots, v^n, -v^1, \ldots, -v^n\}$
        \State $ED(b_0) \gets b_0\_ed$
        \State $\mathcal{E} \gets b_0\_ed$
        \State $new\_base\_cones \gets \{b_0\}$
        \For{$i \in \{1, \ldots, M\}$}
            :
            \State $base\_cones \gets new\_base\_cones$
            \State $new\_base\_cones \gets \emptyset$
            \For{$b \in base\_cones$}
                :
                \State $U_+ \gets \{u \mid u\in ED(b),\ u_i > 0\}$
                \State $U_- \gets \{u \mid u\in ED(b),\ u_i < 0\}$
                \State $U_0 \gets \{u \mid u\in ED(b),\ u_i = 0\}$
                \State $U_{new} \gets \emptyset$
                \For{$u\in U_+$}
                    :
                    \For{$v\in U_-$}
                        :
                        \If{$adjacent(b, u, v)$}
                            :
                            \State $w = v_i \cdot u - u_i\cdot v $
                            \State $U_{new}\gets \{w\}$
                        \EndIf
                    \EndFor
                \EndFor
                \State $b_+ \gets (b_1, \ldots, b_{i-1}, \infty, 0, \ldots, 0)$
                \State $ED(b_+) \gets U_+ \cup U_0 \cup U_{new}$
                \State $b_- \gets (b_1, \ldots, b_{i-1}, -\infty, 0, \ldots, 0)$
                \State $ED(b_-) \gets U_- \cup U_0 \cup -U_{new}$
                \State $\mathcal{E} \gets \mathcal{E} \cup U_{new}$
                \State $new\_base\_cones \gets new\_base\_cones \cup \{b, b_+, b_-\}$
            \EndFor
        \EndFor \\
        \Return $\mathcal{E}$
    \end{algorithmic}
\end{algorithm}

The adjacency criterion is discussed in more detail in \cite{zolotykh2012new},
together with multiple ways that adjacency can be efficiently tested. In short,
the adjacency criterion is needed to ensure that $w$ is actually a new extreme direction
not already present in the old cone. The condition we use, adapted for our setting states that $u, v$ are
adjacent in $C^A_b$ if and only if
$$Z_b(u) \cap Z_b(v) \supset Z_b(w), \forall w\in ED(b)\setminus\{u, v\},$$
where $Z_b(x) = \{i | b_i \neq 0 \land u_i = 0\}.$

    To reduce the number of cones we have to consider, we note that if for some
cone $C^A_b \in \mathcal{C}_i$ and restriction $x_{i+1} \geq 0,$ the sets $U_+$
or $U_-$ are empty, then the set of new extreme directions $U_{new}$ is also
empty and it follows that $ED(b)\supseteq ED(b_\pm)$. If this is the case
for all restrictions $x_{j} \geq 0,\ j > i$, then no cone built from $C^A_b$
can ever generate a new extreme direction. Since we are only interested in new
extreme directions, we can ignore $C^A_b$ and any cone built from it in our
calculations. This allows the following optimization: We only add a cone $b$ (or
$b_+$, $b_-$) to $new\_base\_cones$ if:
$$\exists j > i:\ ED(b)\cap \{x |  x_j > 0\} \neq \emptyset \land ED(b)\cap \{x |  x_j < 0\} \neq \emptyset.$$

Note that this optimization only checks whether a pair $(u, v)$ exists, not
whether it is also \textit{adjacent}. This is because this property would not
extend to all cones built from $C^A_b$ in the same way. Since adjacency of $u$
and $v$ depends on all other extreme directions, it is possible that the current
cone has no adjacent pairs $ (u, v), u\in U_+, v\in U_-$ for any restriction,
but has a pair $(u, v)$ that is adjacent in a future cone $C^A_{b'}$ after
another restriction has removed extreme directions from $C^A_b$.
\end{toappendix}

\subsection{Cylindrical Algebraic Decomposition}
Cylindrical algebraic decomposition is an algorithm that partitions a subset of
$\mathbb{R}^n$ specified by a set of polynomials into so-called \emph{cells}
(i.e., connected subsets of $\mathbb{R}^n$). For a complete description of the
algorithm, see \citeay{lavalle2006planning}. As implemented in
\textit{Mathematica} \cite{cylindersmathworld}, it allows us to solve a system
of polynomial equations and inequalities. Each cell successively assigns to each
variable $v_i$ an algebraic function specifying an interval $I_i(v_1, \ldots,
v_{i-1})$ of possible values for that variable. To extract a specific solution,
we can first select $v_1\in I_1$, then $v_2\in I_2(v_1)$, then $v_3\in I_3(v_1,
v_2)$, and so on.

Note that cylindrical algebraic decomposition is only defined for polynomials
with coefficients in $\mathbb Q$. Strictly speaking, our approach is therefore
restricted to games with rational payoffs.

\begin{figure}
\newcommand{\myfrac}[2]{\frac{#1}{#2}}
\begin{tabular}{l|l|l|l|l|l}
$\beta(a)$ & $\beta(c)$ & $\beta(e)$ & $\mu(\left<a\right>)$ & $\mu(\left<a,d\right>)$ & $U^E_2(\left<\right>)$ \\ \hline
$\in [0,\myfrac{2}{3})$ & $= 0$ & $= 0$ & $= \beta(a)$ & $= \beta(a)$ & $= 2 - \beta(a)$ \\
$= \myfrac{2}{3}$ & $= 0$ & $\in [0, 1]$ & $= \myfrac{2}{3}$ & $= \myfrac{2}{3}$ & $= \myfrac{4}{3}$ \\
$\in (\myfrac{2}{3},1]$ & $= 0$ & $= 1$ & $= \beta(a)$ & $= \beta(a)$ & $= 2 \cdot \beta(a)$
\end{tabular}

\caption{All sequential equilibria of the game from Figure~\ref{fig:localsr} as
returned by the cylindrical algebraic decomposition. Note that, e.g., $\beta(b)
= 1-\beta(a)$, $\mu(\left<b\right>) = 1-\mu(\left<a\right>)$ and
$U_1^E(\left<\right>) = 0$.}
\label{fig:solution}
\end{figure}

Cylindrical algebraic decomposition has a computational complexity that is
double exponential in the number of variables. This makes it infeasible for
large games. Our system of polynomials contains one variable $\beta(I)(a)$ for
each action and one variable $\mu(I)(h)$ for each history in each information
set. We can reduce the number of variables by defining one belief and one action
probability for each information set implicitly through the others.

Another optimization that has proven effective in some cases is to add
additional equations describing the Nash equilibria of the game. Since all
sequential equilibria are Nash equilibria, this does not change the solution
space and we obtain the same solutions. The polynomial equations characterizing
the Nash equilibria follow the usual idea of restricting the behavioral
strategies to best responses.

An example output can be seen in Figure~\ref{fig:solution}. The sequential
equilibria are partitioned into three cells with infinitely many elements.

\section{Conclusion}\label{sec:conclusion}

The main contribution of this paper is to show that the set of all sequential
equilibria of an extensive-form game can be represented by a system of
polynomial equations and inequalities. To this end, we combined theoretical
results by \citeay{hendon1996one} and \citeay{kohlberg1997independence}. We
furthermore described how to obtain and solve this system using symbolic
computation.

For sequential rationality, we used a local form of sequential rationality that
considers only the deviations at each information set. Local sequential
rationality holds at $I$ if and only if the acting player believes all played
actions to be best responses to the other players' strategies. We have shown how
to represent this property as a system of equations and inequalities that are
linear in the beliefs $\mu(I)(h)$ and polynomial in the strategies
$\beta(I)(a)$. Similar to the one-shot-deviation principle for subgame perfect
equilibria, local sequential rationality is a sufficient (and necessary)
condition for sequential rationality under the assumption that the assessment
is already consistent, as shown by \citeay{hendon1996one}.

For consistency, we briefly discussed the
restrictions that consistency can pose on unreached information sets, which are
not entirely obvious. \citeay{kohlberg1997independence} proposed a way to
represent consistency as a system of polynomial equations.
We have elaborated their approach in sufficient detail to be easily
\mbox{implemented}.
While this gives us a solution to represent consistency, it comes with a new
problem: finding the extreme directions of a set of cones.

The double description method provides an established method to calculate the
extreme directions of a single cone. However, it is inefficient for calculating
the extreme directions of a set of cones. We proposed a modified version of the
algorithm, that takes advantage of its iterative nature to avoid calculating the
extreme directions of the same cone multiple times. Our method still considers
each cone individually, and even though we found a way to prune some cones, the
number of cones is still prohibitively large for larger games.

We use symbolic computation to solve our system of polynomial equations and
inequalities. This gives us a compact and exact representation of all sequential
equilibria of a game, even if there are connected components of infinitely many
equilibria.
However, both steps of this process, generating the system of equations and
solving it, are infeasible for large games. In the first step, the number of
cones is exponential in the size of the game tree, and in the second step, the
cylindrical algebraic decomposition is double exponential in the number of
variables in an assessment. While we have proposed some minor optimizations,
pruning some of the cones and substituting some of the variables, these
improvements are not sufficient to make the approach feasible for larger games.

Nevertheless, our implementation
successfully handles small games such as Selten's Horse or small signaling
games. It is already difficult to find all sequential equilibria for these games
by hand. To the best of our knowledge, our implementation is the first tool that
finds all sequential equilibria in finite imperfect information games,
representing them symbolically.
Besides for analyzing small games, we see a practical use for teaching the
concept of sequential equilibria in game theory courses.

There are approaches from related work for computing sequential equilibria for
restricted subsets of games.
\citeay{sequentialtwoplayer} use minimax strategies to compute sequential
equilibria for two-player games.
\citeay{largesequential} compute sequential equilibria in a class of games
where chance nodes are the only source of uncertainty.
\citeay{turocy2010computing} provides a numerical algorithm to compute
sequential equilibria for finite imperfect information games, implemented in
Gambit. However, the implementation suffers from issues of numerical instability
making it unreliable \cite{sequentialtwoplayer}.
Finally, \citeay{panozzo2014algorithms} proposed some approaches for the
algorithmic verification of sequential equilibria.

In future work, we plan to investigate techniques to improve the performance of
our approach. These could include improving the computation of the extreme
directions of the set of cones, or directly optimizing the equations to gain
performance in the cylindrical algebraic decomposition. There are also
alternative characterizations of consistency that we could potentially use that
do not rely on the extreme direction of polyhedral cones
\cite{streufert2012additive,dilme2022lexicographic}.
We believe that this may allow us to solve slightly larger games such as Kuhn
Poker \cite{kuhn1950simplified}.

\begin{acks}
  This work was supported by the ANR LabEx CIMI (grant ANR-11-LABX-0040) within
  the French State Programme ``Investissements d’Avenir'', and the EU ICT-48
  2020 project TAILOR (No. 952215). We thank Robert Mattmüller, Umberto Grandi,
  Bernhard von Stengel, and Martin Antonov for their helpful feedback.
\end{acks}

\bibliographystyle{ACM-Reference-Format}
\bibliography{bibliography.bib}
\appendix

\end{document}